\documentclass[onecolumn, 11pt, letterpaper]{article}

\usepackage[margin = 1in]{geometry}

\usepackage{graphicx}
\usepackage{amsthm}
\usepackage{amsmath,amssymb,amsfonts,bm,setspace}\allowdisplaybreaks
\usepackage[utf8]{inputenc}
\usepackage{tikz}
\usepackage{algorithm,algorithmicx}
\usepackage[noend]{algpseudocode}
\usepackage{fancyhdr}
\usepackage{autobreak}
\usepackage{mathtools}
\usepackage{mathrsfs}
\usepackage{enumerate}
\usepackage{xifthen}
\usetikzlibrary{decorations.pathreplacing,decorations.markings,decorations.pathmorphing,decorations.shapes,arrows.meta,positioning}

\usepackage[hidelinks]{hyperref}

\DeclareMathOperator*{\argmax}{argmax}

\newcommand{\compilehidecomments}{false}

\ifthenelse{ \equal{\compilehidecomments}{true} }{%
	\newcommand{\xiaoming}[1]{}
	\newcommand{\zhijie}[1]{}
	\newcommand{\jialin}[1]{}
	\newcommand{\shuo}[1]{}
}{
	\newcommand{\xiaoming}[1]{{\color{blue!50!black}  [\text{Xiaoming:} #1]}}
	\newcommand{\zhijie}[1]{{\color{red!60!black} [\text{Zhijie:} #1]}}
	\newcommand{\jialin}[1]{{\color{brown!60!black} [\text{Jialin:} #1]}}
	\newcommand{\shuo}[1]{{\color{green!60!black} [\text{Shuo:} #1]}}
}

\newtheorem{lemma}{Lemma}
\newtheorem{theorem}{Theorem}
\newtheorem{definition}{Definition}

\newcommand*\samethanks[1][\value{footnote}]{\footnotemark[#1]}

\begin{document}
\title{Improved Deterministic Streaming Algorithms for Non-monotone Submodular Maximization}

\author{
	Xiaoming Sun\thanks{Institute of Computing Technology, Chinese Academy of Sciences, Beijing, China, and School of Computer Science and Technology, University of Chinese Academy of Sciences, Beijing, China.
E-mail:   
\{sunxiaoming, zhangjialin, zhangshuo19z\}@ict.ac.cn}
\and
Jialin Zhang\samethanks
\and
Shuo Zhang\samethanks
}

\maketitle              

\begin{abstract}
Submodular maximization is one of the central topics in combinatorial optimization. It has found numerous applications in the real world. 
Streaming algorithms for submodule maximization have gained attention in recent years, allowing for real-time processing of large data sets by looking at each piece of data only once.
However, most of the state-of-the-art algorithms are subject to monotone cardinality constraint. 
There remain non-negligible gaps with respect to approximation ratios between cardinality and other constraints like $d$-knapsack in non-monotone submodular maximization. 

In this paper, we propose deterministic algorithms with improved approximation ratios for non-monotone submodular maximization. 
Specifically, for the cardinality constraint, we provide a deterministic $1/6-\epsilon$ approximation algorithm with $O(\frac{k\log k}{\epsilon})$ memory and sublinear query time, while the previous best algorithm is randomized with a $0.1921$ approximation ratio. 
To the best of our knowledge, this is the first deterministic streaming algorithm for the cardinality constraint.
For the $d$-knapsack constraint, we provide a deterministic $\frac{1}{4(d+1)}-\epsilon$ approximation algorithm with $O(\frac{b\log b}{\epsilon})$ memory and $O(\frac{\log b}{\epsilon})$ query time per element. 
To the best of our knowledge, there is currently no streaming algorithm for this constraint.

\end{abstract}

\section{Introduction}

Submodular maximization has become one of the central problems in optimization and combinatorics, with submodular functions embodying the essence of diminishing returns. 
The inherent trait of a submodular function—that the marginal value of an item decreases as one adds more items to the set—makes it a particularly intriguing function to maximize.
Submodular maximization problem arises in a variety of applications including influence maximization \cite{KKT03}, machine learning \cite{DK08,KED+17}, sensor placement and feature selection \cite{IB12,IB13}, and information gathering \cite{KG11}.

Due to its remarkable significance, submodular maximization has been studied over the past forty years, especially in the offline model. For the monotone case, it cannot exceed $1-1/e$ under cardinality constraints \cite{NemhauserW78,Fei98}, a tight approximation ratio met by the natural greedy algorithm \cite{NemhauserWF78}. When combined with enumeration techniques, greedy algorithm also achieves a $1-1/e$ approximation ratio for knapsack constraints \cite{SAJ99,Sviridenko04}. 
Continuous greedy algorithm \cite{Vondrak08}, proposed by Vondr{\'{a}}k in 2008, generalizes to address various constraints, including matroid, and $d$-knapsack and so on constraints with down-closed properties, necessitating randomization due to its sampling process \cite{KulikST09,FeldmanNS11}. 
For the non-monotone case, where randomization is a tool, the continuous greedy technique facilitates a $0.401$ approximation \cite{BuchbinderF23}, trailing behind the $0.491$ limit under cardinality constraints but being superior to the deterministic $1/e$ for the cardinality and $1/4$ for the knapsack constraint \cite{BuchbinderF18,SunZZZ22}. 

However, as the digital era progresses, the real challenge arises when we have to deal with massive datasets that are constantly evolving. 
This is where streaming algorithms come into play. 
Traditional optimization algorithms that work on static, small-scale datasets fall short when applied to large streaming data, which is continuously fed and cannot be stored or processed in its entirety. 
Streaming algorithms, in contrast, allow for processing massive datasets in real-time, by seeing each piece of data only once or a few times, making only a fraction of the entire data available at any point. 
Badanidiyuru et al.~\cite{BadanidiyuruMKK14} propose a deterministic $1/2-\epsilon$ approximation algorithm for monotone submodular maximization with cardinality constraints in streaming model with $O(\frac{k\log k}{\epsilon})$ memory and sublinear time.
Feldman et al.~\cite{FeldmanNSZ20} show that the problem can not be approximated within a ratio better than $1/2$ even under the cardinality constraint under the streaming model, whether the function is monotone or not. 
We can see that streaming algorithms are more difficult than offline algorithms because the information mastered is not the whole picture, but only a part of them.
For non-monotone case, Feldman et al.~\cite{FeldmanFSZ21} propose a randomized streaming algorithm with $0.1921$ approximation.
For $d$-knapsack constraint, deterministic streaming algorithms with $\frac{1}{2d+1}-\epsilon$ approximation ratio for monotone submodular maximization are proposed in \cite{KumarMVV15,YuXC16}. When the function is non-monotone, to the best of our knowledge, there is currently no streaming algorithm for this constraint.

The motivation to study submodular maximization with streaming algorithms under $d$-knapsack constraints arises from both practical and theoretical considerations.
\textbf{Theoretical Viewpoint}: It is an interesting and important question whether randomness is essentially necessary in submodular maximization.
\textbf{Practical Viewpoint}: As data grows, the need for efficient algorithms that can handle such large amounts of data becomes critical.
In addition, random algorithms only work in the average case, while deterministic algorithms still work in the worst case.
By exploring the deterministic streaming algorithms of submodular maximization problem in different constraints, we hope to find solutions relevant to today's data-driven world, while also advancing theoretical knowledge.

\subsection{Our Contribution}

In this paper, we provide several improved deterministic streaming algorithms for \textbf{non-monotone} submodular maximization subject to different constraints. 
\begin{itemize}
	\item For the cardinality constraint, we present a deterministic $1/6-\epsilon$ approximation algorithm. 
    It does 1 pass over the data set, stores at most $O(\frac{k\log k}{\epsilon})$ elements and has $O(\frac{\log k}{\epsilon})$ query complexity per element, where $k$ is the maximum size of the constrained set.
    To the best of our knowledge, this is the first deterministic streaming algorithm for the constraint.
	\item For the $d$-knapsack constraint, we propose a deterministic streaming     algorithm with $\frac{1}{4(d+1)}-\epsilon$ approximation. 
    It does 1 pass over the data set, stores at most $O(\frac{ b\log b}{\epsilon})$ elements and has $O(\frac{\log b}{\epsilon})$ query complexity per element, where $b$ is the standardized capacity of the $d$-knapsack.        
    To the best of our knowledge, there is currently no streaming algorithm for this constraint.
	
\end{itemize}
We make a more detailed comparison between our and previous results in Table \ref{tab: main}.

 \begin{table}[ht]
	 \centering
	 \begin{tabular}{cccccc}
            \hline
		  Constraint & Reference & Ratio & Memory & Monotonicity & Type \\ 
		  \hline\hline
		  Cardinality & Badanidiyuru et al.~\cite{BadanidiyuruMKK14} & $1/2-\epsilon$ & $\mathcal{O}(\frac{k\log k}{\epsilon})$ & Monotone & Det \\
		  Cardinality & Feldman et al.~\cite{FeldmanFSZ21} & $0.1921$ & $\mathcal{\widetilde{O}}(k)$ & Non-monotone & Rand \\
		  \hline
   
		  Cardinality & Theorem \ref{thm: Streaming Repeat Greedy Cardinality 3} & $1/6-\epsilon$ & $\mathcal{O}(\frac{k\log k}{\epsilon})$ & Non-monotone & Det \\ 
		  \hline
		  \hline
		  $d$-knapsack & Qilian Yu et al.~\cite{KumarMVV15} & $\frac{1}{2d+1}-\epsilon$ & $\mathcal{O}(\frac{b\log b}{\epsilon})$ & Monotone & Det \\
		  \hline
		  $d$-knapsack & Theorem \ref{thm: Streaming for d-Knapsack} & $\frac{1}{4(d+1)}-\epsilon$ & $\mathcal{O}(\frac{b\log b}{\epsilon})$ & Non-monotone & Det \\ 
		  \hline
		 \end{tabular}
	 \vspace{2mm}
	 \caption{Streaming algorithms for non-monotone submodular maximization under a cardinality, and a $d$-knapsack constraint. 
        ``Memory'' refers to the number of elements in the register.
        ``Rand'' is short for ``Randomized'' and ``Det'' is short for ``Deterministic''.
        }
	 \label{tab: main}
	 \end{table}

\subsection{Related Work}

To better illustrate the improvement of our results, this subsection provides a list of results in the literature on \emph{non-monotone} submodular maximization under a cardinality constraint, and a $d$-knapsack constraint.

For the cardinality constraint, the best semi-streaming algorithm is based on the offline continuous greedy technique and achieves a $0.2779$ approximation ratio \cite{AlalufEFNS20}.
However, this algorithm suffers from high query complexity and is randomized by using the continuous greedy.
There exist different randomized algorithms that achieve $0.1921$ approximation ratio using $\mathcal{\widetilde{O}}(k)$ memory \cite{FeldmanFSZ21}, where $k$ is the maximum size of the set.
When randomness is not allowed, there is no streaming algorithm for non-monotone submodular maximization problem with a constant approximation ratio. For the monotone case, Badanidiyuru et al.~\cite{BadanidiyuruMKK14} propose a deterministic streaming algorithm with $1/2-\epsilon$ approximation ratio, using $O(\frac{k\log k}{\epsilon})$ memory and $O(\frac{\log k}{\epsilon})$ query time per element.

For the $d$-knapsack constraint, when $d$ is constant or the width of the constraints is large, constant approximation is possible \cite{LMNS10,SMM11,BuchbinderF19}.
Currently, the best known algorithm is again based on the continuous greedy algorithm and has an approximation ratio of $0.401$ \cite{BuchbinderF23}.
Sun et al.~\cite{SunZZZ23} propose a deterministic algorithm for the offline setting with $1/6$ approximation when the width of the knapsack is large enough.
To the best of our knowledge, there is no streaming algorithm for the problem in the literature.
For the monotone case, Yu et al.~\cite{YuXC16} propose a deterministic streaming algorithm for $d$-knapsack with $\frac{1}{2d+1}-\epsilon$ approximation ratio, whose memory cost is $O(\frac{b\log b}{\epsilon})$ and query time is $O(\frac{\log b}{\epsilon})$ per element.

\subsection{Organization}

In Section \ref{sec: pre}, we formally introduce the problem of non-monotone submodular maximization under a cardinality constraint and a $d$-knapsack constraint. 
In Section \ref{sec: cardinality}, we propose a deterministic streaming algorithm for the cardinality constraint with $\frac{1}{6}-\epsilon$ approximation ratio using $O(\frac{k\log k}{\epsilon})$ memory.
In Section \ref{sec: d-knapsack}, we propose deterministic streaming algorithms for the $d$-knapsack constraint with $\frac{1}{4(d+1)}-\epsilon$ approximation ratios using $O(\frac{b\log b}{\epsilon})$ memory. To the best of our knowledge, this is the first streaming algorithm for this problem.
In Section \ref{sec: conclusion}, we conclude the paper and list some future directions.

\section{Preliminaries}
\label{sec: pre}

In this section, we state the problems studied in this paper.

\begin{definition}[Submodular Function]
	Given a finite ground set $N$ of $n$ elements, a set function $f:2^N\mapsto \mathbb{R}$ is submodular if for all $S,T\subseteq N$,
	\[ f(S)+f(T) \geq f(S\cup T)+ f(S\cap T). \]
	Equivalently, $f$ is submodular if for all $S\subseteq T\subseteq N$ and $u\in N\setminus T$,
	\[ f(S\cup\{u\})-f(S)\geq f(T\cup\{u\})-f(T). \]
\end{definition}
For convenience, we use $f(u)$ to denote $f(\{u\})$, $f(S+u)$ to denote $f(S\cup\{u\})$, $f(u\mid T)$ to denote the marginal value $f(T+u)-f(T)$ of $u$ with respect to $T$, and $f(S\mid T)$ to denote the marginal value $f(S\cup T)-f(T)$.
The function $f$ is \emph{non-negative} if $f(S)\geq 0$ for all $S\subseteq N$.
The function $f$ is \emph{monotone} if $f(S)\leq f(T)$ for all $S\subseteq T\subseteq N$.

\begin{definition}[Cardinality]
    Given a finite ground set $N$, and a integer $k$.
    For set $S\subseteq N$, the cardinality constraints can be defined as $\mathcal{I}=\{S\mid |S|\le k\}$.
\end{definition}

\begin{definition}[Knapsack]
	Given a finite ground set $N$, assume there is a budget $B$, and each element $u\in N$ is associated with a cost $c(u)>0$.
	For set $S\subseteq N$, its cost $c(S)=\sum_{u\in S}c(u)$.
	We say $S$ is feasible if $c(S)\leq B$.
	The knapsack can be written as $\mathcal{I}=\{S\mid c(S)\leq B\}$.
\end{definition}
If $c(u)=1$ for all $u\in N$ in the knapsack and let $B=k$, the knapsack reduces to the \emph{cardinality constraint}.

\begin{definition}[$d$-Knapsack]
Given a finite ground set $N$, a matrix $\textbf{C}\in (0,\infty)^{d\times n}$, and a vector $\textbf{b}\in (0,\infty)^d$. For set $S \subseteq N $, the $d$-knapsack constraints can be written as $\mathcal{I}=\{S\mid \textbf{Cx}_\textbf{S}\leq \textbf{b}\}$, where $\textbf{x}_\textbf{S}$ stands for the characteristic vector of the set $S$.
\end{definition}

Without loss of generality, for $1\le i\le d$, $1\le j\le n$, we assume that $c_{i,j}\le b_i$. 
For the sack of simplicity, we can standardize the constrained problem. 
Let $b'=\max_{1\le i\le d}b_i$, and $c'=\min_{1\le i\le d,1\le j\le n} bc_{i,j}/b_i$. For $1\le i\le d, 1\le j\le n$, we replace each $c_{i,j}$ with $b'c_{i,j}/b_ic'$, and $b_i$ with $b'/c'$. The standardized problem has the same optimal solution and $c_{i,j}\ge 1$.

In the rest of the paper, we only consider the standardized version of the $d$-knapsack constrained submodular maximization problem. 
Let $b=b'/c'$ such that the $d$-knapsack has equal capacity $b$ in all dimensions, and the cost of each element $c_{i,j}\ge 1$, for $1\le i\le d$, and $1\le j\le n$.

When $d=1$, the $d$-knapsack constraint reduce to the knapsack constraint.

\begin{definition}[Constrained Submodular Maximization]
	The constrained submodular maximization problem has the form
	\[ \max\{f(S)\mid S\in\mathcal{I}\}. \]
\end{definition}
In this paper, the constraint $\mathcal{I}$ is assumed to be the cardinality constraint, or the $d$-knapsack constraint respectively. The objective function $f$ is assumed to be non-negative, non-monotone, and submodular.
Besides, $f$ is accessed by a value oracle that returns $f(S)$ when $S$ is queried.
The efficiency of any algorithm for the problem is measured by the number of queries it uses.

\section{Streaming Algorithm for Cardinality}
\label{sec: cardinality}
In this section, we present a deterministic streaming algorithm for submodular maximization under a cardinality constraint.
Our algorithm is obtained by combining the Sieve-Streaming algorithm \cite{BadanidiyuruMKK14} for the monotone case and the technique from \cite{GuptaRST10} for dealing with the lack of monotonicity.
First, we assume that we have some knowledge of OPT, and then remove this assumption by estimating OPT based on the max value of a single element. Finally, we remove all the assumptions and propose a 1-pass streaming algorithm subject to a cardinality constraint.

\subsection{Streaming Algorithm for Cardinality Knowing OPT}
Suppose that we have a value $v$ such that $\alpha $OPT$ \le v\le $ OPT, for some $\alpha \in (0,1]$. We construct the algorithm to choose elements by the threshold according to the value of $v$. 

\begin{algorithm}[ht]
	\caption{Streaming Repeat Greedy for Cardinality Knowing OPT}
	\begin{algorithmic}[1]
		\State \textbf{Input} $v$ such that $\alpha $OPT$ \le v\le $ OPT, for some $\alpha \in (0,1]$.
		\State $S_1, S_2, S_3 = \emptyset$.
            \State $\tau = v / 6$.
		\Comment{$\tau$ is the threshold. }  
		\For{$j=1$ to $n$}
            \Comment{$u_j$ is the $j$-th element in the dataset. } 
            \If{$f(u_j\mid S_1)\ge \frac{\tau}{k}  $ and $|S_1|<k $} 
            \State $S_1= S_1\cup\{u_j\}$.
            \ElsIf {$f(u_j\mid S_2)\ge \frac{\tau}{k}  $ and $|S_2|<k $} 
            \State $S_2= S_2\cup\{u_j\}$.
            \EndIf
		\EndFor
            \State $S_3 = $ \textbf{Unconstrained} ($S_1$).
		\State\Return $S= \argmax\{f(S_1),f(S_2),f(S_3)\}$.
	\end{algorithmic}
	\label{alg: Streaming Repeat Greedy Cardinality 1}
\end{algorithm}

\begin{theorem}
    Assuming that the input $v$ satisfies $\alpha $OPT$ \le v\le $ OPT, Algorithm \ref{alg: Streaming Repeat Greedy Cardinality 1} satisfies the following properties, where the unconstrained submodular maximization is solved by a deterministic $1/2$-approximation algorithm proposed by \cite{BuchbinderF18}.
    \begin{itemize}
         \item It outputs a set $S$ such that $|S|\le k$ and $f(S)\ge \frac{\alpha}{6}$OPT.
         \item It does 1 pass over the data set, stores at most $k$ elements and has $O(1)$ query complexity per element.
    \end{itemize}
    \label{thm: Streaming for Cardinality Knowing OPT}
\end{theorem}

\begin{proof}
    We prove this theorem by two case subject to the size of two candidate sets $S_1$ and $S_2$ at the end of streaming.
    Let set $O$ be the optimal set such that $f(O)=\text{OPT}$, and $S_l^j$ be the candidate set at the end of the $j$-iteration, $S_l^0=\emptyset$, and $S_l^n=S_l$, $l=1,2$.
    \begin{itemize}
        \item Case 1: At the end of the algorithm, at least one candidate set has $k$ elements. Without lose of generation, let $ |S_1| = k $, and for $1\le i \le k$, let $u_i$ be the $i$-th element added into $S_1$. Then we have that $f(S) = \sum_{i=1}^k (f(\{u_1, u_2, \dots , u_i\})- f(\{u_1, u_2, \dots , u_{i-1}\}) ) \ge k\cdot \frac{\tau}{k} = \frac{v}{6} $.
        Thus, $f(S) \ge f(S_1) \ge \frac{\alpha }{6}$OPT.
        
        \item Case 2: Neither candidate set is full, that is, $|S_1|<k$ and $|S_2|<k$ at the end of the algorithm.
        For every element $u\in O\setminus S_1$, let it be rejected in iteration $j$. 
        Then we have $f(u\mid S_1)\le f(u\mid S_1^j)\le \frac{\tau}{k}$, which implies that $f(O\mid S_1) \le \sum _{u\in O\setminus S_1} f(u\mid S_1) \le \tau$, and further $f(S_1)\ge f(S_1\cup O) - \tau $.
        Similarly, we have $f((O\setminus S_1) \mid S_2 ) \le \tau$, which implies that $f(S_2)\ge f(S_2 \cup (O\setminus S_1)) - \tau $.
        \begin{align*}
            f(S) &\ge \frac{1}{2} \left(f(S_1) + f(S_2)\right) \\
            &\ge \frac{1}{2} \left(f(S_1\cup O) + f(S_2 \cup (O\setminus S_1)) - 2\tau \right) \\
            &= \frac{1}{2} \left(f(S_1\cup O) + f(S_2 \cup (O\setminus S_1)) + f(S_1 \cap O) \right) -\frac{1}{2}f(S_1 \cap O) -\tau \\ 
            &\ge \frac{1}{2}f(O) -\frac{1}{2}f(S_1 \cap O) -\tau.
        \end{align*}
        The last inequality is due to submodularity and non-negative of $f$, notice that $f(S_1\cup O) + f(S_2 \cup (O\setminus S_1)) + f(S_1 \cap O)\ge f(S_1\cup O) + f(S_2 \cup O) \ge f(O) + f(S_1\cup S_2 \cup O) \ge f(O)$.
        We use a $\frac{1}{2}$-approximation deterministic unconstrained submodular maximization algorithm to solve the part of $f(S_1 \cap O)$ \cite{BuchbinderF18}.
        That is, if $f(S_1 \cap O)\ge 2\tau $, then we have $f(S)\ge f(S_3)\ge \tau $. 
        Otherwise, $f(S_1\cap O) < 2\tau $, then
        \begin{align*}
            f(S) \ge \frac{1}{2}f(O) -\frac{1}{2}f(S_1 \cap O) -\tau 
            \ge \frac{1}{2}f(O) -2\tau
            \ge 3\tau -2\tau 
            =\tau.
        \end{align*}
        Where the last inequality is by $f(O)\ge v\ge 6\tau$.
        Then we have $f(S) \ge \frac{\alpha }{6} \text{OPT}$ in any case.
    \end{itemize}
\end{proof}

\subsection{Streaming Algorithm for Cardinality Knowing the Max}

We use the maximum value element $m=\max_{u\in N} f(u)$ to estimate OPT. By submodularity, we have that $m\le \text{OPT} \le k\cdot m$. Consider the following set $Q=\{(1+\epsilon)^i\mid i\in \mathbb{Z}, m\le (1+\epsilon)^i \le k\cdot m\}$. At least one of the value $v\in Q$ should be a pretty good estimate of OPT, such that $\frac{1}{1+\epsilon}\text{OPT}\le v \le \text{OPT}$.
We can run Algorithm \ref{alg: Streaming Repeat Greedy Cardinality 1} once for each value $v\in Q$ in parallel, producing one candidate set for each $v\in Q$. As final output, we return the best solution obtained.

\begin{algorithm}[ht]
	\caption{Streaming Repeat Greedy for Cardinality Knowing the Max}
	\begin{algorithmic}[1]
        \State \textbf{Input} $m=\max_{u\in N} f(u)$.
        \State $Q=\{(1+\epsilon)^i\mid m\le (1+\epsilon)^i \le k\cdot m\}$.
        \For{$v\in Q$}
        \State $S^v_1, S^v_2, S^v_3 = \emptyset$.
        \EndFor
	\For{$j=1$ to $n$}
        \Comment{$u_j$ is the $j$-th element in the dataset. } 
        \For{$v\in Q$}
        \State $\tau = v / 6$.
        \Comment{$\tau$ is the threshold. }  
        \If{$f(u_j\mid S_1^v)\ge \frac{\tau}{k}  $ and $|S_1^v|<k $} 
        \State $S_1^v= S_1^v\cup\{u_j\}$.
        \ElsIf {$f(u_j\mid S_2^v)\ge \frac{\tau}{k}  $ and $|S_2^v|<k $} 
        \State $S_2^v= S_2^v\cup\{u_j\}$.
        \EndIf
        \EndFor
	\EndFor
        \For{$v\in Q$}
        \State $S_3^v = $ \textbf{Unconstrained} ($S_1^v$).
        \EndFor
	\State\Return $S= \argmax_{v\in Q}\{f(S_1^v),f(S_2^v),f(S_3^v)\}$.
	\end{algorithmic}
	\label{alg: Streaming Repeat Greedy Cardinality 2}
\end{algorithm}

\begin{theorem}
    Assuming that the input $m$ satisfies $m=\max_{u\in N}f(u)$, Algorithm \ref{alg: Streaming Repeat Greedy Cardinality 2} satisfies the following properties, where the unconstrained submodular maximization is solved by a deterministic $1/2$-approximation algorithm proposed by \cite{BuchbinderF18}.
    \begin{itemize}
         \item It outputs a set $S$ such that $|S|\le k$ and $f(S)\ge (\frac{1}{6}-\epsilon)$OPT.
         \item It does 1 pass over the data set, stores at most $O(\frac{k\log k}{\epsilon})$ elements and has $O(\frac{\log k}{\epsilon})$ query complexity per element.
    \end{itemize}
    \label{thm: Streaming Repeat Greedy Cardinality 2}
\end{theorem}

\begin{proof}
    Because that at least one of the value $v\in Q$ should be a pretty good estimation of OPT, such that $\frac{1}{1+\epsilon}\text{OPT}\le v \le \text{OPT}$.
    The approximation ratio of Algorithm \ref{alg: Streaming Repeat Greedy Cardinality 2} is the same as that of Algorithm \ref{alg: Streaming Repeat Greedy Cardinality 1}, together with their proofs. We only analyze the memory and the query complexity of Algorithm \ref{alg: Streaming Repeat Greedy Cardinality 2}.
    Note that $|Q|=O(\frac{\log k}{\epsilon})$, we keep track of $O(\frac{\log k}{\epsilon})$ many sets $S_l^v$ of size at most $k$ each, bounding the size of memory by $O(\frac{k\log k}{\epsilon})$, $l=1,2,3$. Moreover, the query time is $O(\frac{\log k}{\epsilon})$ per element.
\end{proof}

\subsection{1-Pass Streaming Algorithm for Cardinality}

In order to find the maximum value element $m=\max_{u\in N} f(u)$ during the streaming, we modify the set into $Q=\{(1+\epsilon)^i\mid i\in \mathbb{Z}, m\le (1+\epsilon)^i \le 6k\cdot m\}$. 
It can be seen that when a threshold $v$ is instantiated from the set $Q$, every elemnt with marginal value $\frac{v}{6k}$ to $S^v$ will appear on or after $v$ is instantiated.

\begin{algorithm}[ht]
	\caption{Streaming Repeat Greedy for Cardinality}
	\begin{algorithmic}[1]
		\State $Q=\{(1+\epsilon)^i\mid i\in \mathbb{Z}\}$.
            \For{$v\in Q$}
            \State $S^v_1, S^v_2, S^v_3 = \emptyset$.
            \EndFor
            \State $m=0$.
		\For{$j=1$ to $n$}
            \Comment{$u_j$ is the $j$-th element in the dataset. } 
            \State $m=\max(m, f(u_j))$.
            \State $Q_j=\{(1+\epsilon)^i\mid m\le (1+\epsilon)^i \le 6k\cdot m\}$.
            \State Delete all $S_v$ such that $v\notin Q_j$.
            \For{$v\in Q_j$}
            \State $\tau = v / 6$.
            \Comment{$\tau$ is the threshold. }  
            \If{$f(u_j\mid S_1^v)\ge \frac{\tau}{k}  $ and $|S_1^v|<k $} 
            \State $S_1^v= S_1^v\cup\{u_j\}$.
            \ElsIf {$f(u_j\mid S_2^v)\ge \frac{\tau}{k}  $ and $|S_2^v|<k $} 
            \State $S_2^v= S_2^v\cup\{u_j\}$.
            \EndIf
            \EndFor
		\EndFor
            \For{$v\in Q_n$}
            \State $S_3^v = $ \textbf{Unconstrained} ($S_1^v$).
            \EndFor
		\State\Return $S= \argmax_{v\in Q_n}\{f(S_1^v),f(S_2^v),f(S_3^v)\}$.
	\end{algorithmic}
	\label{alg: Streaming Repeat Greedy Cardinality 3}
\end{algorithm}

\begin{theorem}
    Algorithm \ref{alg: Streaming Repeat Greedy Cardinality 3} satisfies the following properties, where the unconstrained submodular maximization is solved by a deterministic $1/2$-approximation algorithm proposed by \cite{BuchbinderF18}.
    \begin{itemize}
         \item It outputs a set $S$ such that $|S|\le k$ and $f(S)\ge (\frac{1}{6}-\epsilon)$OPT.
         \item It does 1 pass over the data set, stores at most $O(\frac{k\log k}{\epsilon})$ elements and has $O(\frac{\log k}{\epsilon})$ query complexity per element.
    \end{itemize}
    \label{thm: Streaming Repeat Greedy Cardinality 3}
\end{theorem}

\begin{proof}
    When a threshold $v$ is instantiated from the set $Q$, every element with marginal value $\frac{v}{6}$ to $S^v$ will appear on or after $v$ is instantiated. 
    The remaining proof is similar to Theorem \ref{thm: Streaming Repeat Greedy Cardinality 2}.
\end{proof}

\section{Streaming Algorithm for $d$-Knapsack}
\label{sec: d-knapsack}
For $d$-knapsack constraints, there is a challenge that we can not use the enumerate technique like offline setting to solve the problem caused by the capacity is not sufficient. We combining the algorithm \cite{YuXC16} for the monotone case who consider the case of large single element. Then we use the technique from \cite{GuptaRST10} to maintain two candidate sets for dealing with the lack of monotonicity. 
We also start describing the algorithm from assumption that we know the value of OPT.
Then remove this assumption by estimating OPT based on the maximum marginal unit value of all single element. 
Finally, we remove all the assumptions and propose a 1-pass streaming algorithm subject to the $d$-knapsack constraint.

\subsection{Streaming Algorithm for $d$-Knapsack Knowing OPT}
Suppose we have a value $v$ such that $\alpha $OPT$ \le v\le $ OPT, for some $\alpha \in (0,1]$. Then we develop the algorithm to choose elements by the threshold according to the value of $v$. 
 
\begin{algorithm}[ht]
	\caption{Streaming Repeat Greedy for $d$-Knapsack Knowing OPT}
	\begin{algorithmic}[1]
		\State \textbf{Input} $v$ such that $\alpha $OPT$ \le v\le $ OPT, for some $\alpha \in (0,1]$.
		\State $S_1, S_2, S_3 = \emptyset$.
            \State $\tau = \frac{v}{4(d+1)}$.
		\Comment{$\tau$ is the threshold. }  
		\For{$j=1$ to $n$}
            \Comment{$u_j$ is the $j$-th element in the dataset. } 
            \If{$c_{i,j}\ge \frac{b}{2}$ and $\frac{f(u_j)}{c_{i,j}}\ge \frac{2\tau}{b}$ for some $1\le i\le d$}
            \State $S=\{u_j\}$.
            \State\Return $S$.
            \EndIf            
            \If{$\frac{f(u_j\mid S_1)}{c_{i,j}} \ge \frac{2\tau}{b} $ and $\sum_{l\in S_1\cup \{u_j\}} c_{i,l} \le b $ for all $1\le i\le d$} 
            \State $S_1= S_1\cup\{u_j\}$.
            \ElsIf {$\frac{f(u_j\mid S_2)}{c_{i,j}} \ge \frac{2\tau}{b} $ and $\sum_{l\in S_2\cup \{u_j\}} c_{i,l} \le b $ for all $1\le i\le d$} 
            \State $S_2= S_2\cup\{u_j\}$.
            \EndIf
		\EndFor
            \State $S_3 = $ \textbf{Unconstrained} ($S_1$).
		\State\Return $S= \argmax\{f(S_1),f(S_2),f(S_3)\}$.
	\end{algorithmic}
	\label{alg: Streaming for d-Knapsack Knowing OPT}
\end{algorithm}

\begin{theorem}
    Assuming that the input $v$ satisfies $\alpha $OPT$ \le v\le $ OPT, Algorithm \ref{alg: Streaming for d-Knapsack Knowing OPT} satisfies the following properties, where the unconstrained submodular maximization is solved by a deterministic $1/2$-approximation algorithm proposed by \cite{BuchbinderF18}.
    \begin{itemize}
         \item It outputs a set $S$ such that $Cx_S\le b$ and $f(S)\ge \frac{\alpha}{4(d+1)}\text{OPT}$.
         \item It does 1 pass over the data set, stores at most $b$ elements and has $O(d)$ query complexity per element.
    \end{itemize}
    \label{thm: Streaming for d-Knapsack Knowing OPT}
\end{theorem}

\begin{proof}
    The algorithm will terminate when either we find an element $u_j\in N$ such that 
    $c_{i,j}\ge \frac{b}{2}$ and $\frac{f(u_j)}{c_{i,j}}\ge \frac{2\tau}{b}$ for some $1\le i\le d$,
    or we finish one pass through the dataset.
    Here we define that an element $u_j\in N$ is a \textbf{big element} if it satisfies the condition in \textbf{line 5}.
    We first prove that if we find a big element, the set $S=u_j$ will obtain a good approximate ratio.

    \begin{enumerate}
        \item Assume $N$ has at least one big element. Let $a$ be the first big element that algorithm finds. Then algorithm outputs $S=\{u_j\}$ and terminates. Then by \textbf{line 5}, we have $f(S)=f(u_j)\ge \frac{2\tau}{b} \cdot \frac{b}{2}= \tau$. Thus, output $S$ of Algorithm \ref{alg: Streaming for d-Knapsack Knowing OPT} satisfies $f(S)\ge \frac{v}{4(d+1)} \ge \frac{\alpha}{4(d+1)} \text{OPT}$.
        \item Otherwise, if $N$ has no big element, we discuss the following two cases subject to the size of two candidate sets $S_1$ and $S_2$ at the end of streaming.
    \end{enumerate}
    
    \begin{itemize}
        \item Case 1: At the end of the algorithm, at least one candidate set satisfies $\sum_{u_j\in S_l} c_{i,j}\ge \frac{b}{2}$ for some $1\le i\le d$, $l=\{1,2\}$. 
        
        Without lose of generation, let $ l=1 $. Assume that the elements in $S_1$ is selected in order $\{u_1, u_2, \dots , u_{|S_1|}\}$. 
        We have $f(S)\ge \sum_{j=1}^{|S_1|} (f(\{u_1, u_2, \dots , u_j\}) -  f(\{u_1, u_2, \dots , u_{j-1}\}))$. By the algorithm, we have that $f(\{u_1, u_2, \dots , u_j\}) -  f(\{u_1, u_2, \dots , u_{j-1}\}) \ge \frac{2\tau}{b} c_{i,j}$, for all $1\le i\le d$. Then for such $i$ that $\sum_{u_j\in S_1} c_{i,j}\ge \frac{b}{2}$, we have $f(S)\ge \sum_{j=1}^{|S_1|} \frac{2\tau}{b} c_{i,j} \ge \tau $.
        
        \item Case 2: Both of the sizes of $S_1$ and $S_2$ have that $\sum_{u_j\in S_l} c_{i,j}< \frac{b}{2}$ for all $1\le i\le d$, at the end of the algorithm, $l=\{1,2\}$.
        
        For every element $u_j\in O\setminus S_1$, let it be rejected in iteration $j$. There exists an index $\mu(u_j)$, with $1\le \mu(u_j) \le d$ such that $\frac{f(u_j\mid S_1)}{c_{\mu(u_j), u_j}}\le \frac{f(u_j\mid S_1^j)}{c_{\mu(u_j), u_j}} < \frac{2\tau}{b}$.
        By contradiction, if $\frac{f(u_j\mid S_1^j)}{c_{\mu(u_j), u_j}} \ge  \frac{2\tau}{b}$, since $u_j$ is not a big element and $f$ is submodular, we have $c_{i,j}<b/2$, for $1\le i\le d$. Then $u_j$ can be added into $S_1$, where a contradiction occurs.
        
        Let $Y_i$ be the set containing elements $u_j\in O\setminus S_1$ such that $\mu (u_j)=i$, for  $1\le i\le d$. Then $O\setminus S_1 = \cup _{1\le i \le d} Y_i$.
        We have that $f(Y_i \mid S_1)< \frac{2\tau}{b} \sum_{u_j\in Y_i} c_{\mu (u_j),u_j} < 2\tau $, which implies that $f(O\mid S_1) \le \sum _{i=1}^d f(Y_i \mid S_1) < 2d\tau$ and further $f(S_1)\ge f(S_1\cup O) - 2d\tau $.
        
        Similarly, we have $f((O\setminus S_1) \mid S_2 ) \le 2d\tau$, which implies that $f(S_2)\ge f(S_2 \cup (O\setminus S_1)) - 2d\tau $.
        \begin{align*}
            f(S) &\ge \frac{1}{2} \left(f(S_1) + f(S_2)\right) \\
            &\ge \frac{1}{2} \left(f(S_1\cup O) + f(S_2 \cup (O\setminus S_1)) - 4d\tau \right) \\
            &= \frac{1}{2} \left(f(S_1\cup O) + f(S_2 \cup (O\setminus S_1)) + f(S_1 \cap O) \right) -\frac{1}{2}f(S_1 \cap O) - 2d\tau \\ 
            &\ge \frac{1}{2}f(O) -\frac{1}{2}f(S_1 \cap O) - 2d\tau.
        \end{align*}
        Where the third inequality is due to submodularity and non-negative of $f$.
        Notice that $f(S_1\cup O) + f(S_2 \cup (O\setminus S_1)) + f(S_1 \cap O)\ge f(S_1\cup O) + f(S_2 \cup O) \ge f(O) + f(S_1\cup S_2 \cup O) \ge f(O)$.
        We bound the value of $f(S_1 \cap O)$ by two case.
        That is, if $f(S_1 \cap O)\ge \frac{v}{2} - 2d\tau$, then we have $f(S)\ge f(S_3)\ge \frac{v}{4}-d\tau \ge \tau $ by the unconstrained submodular maximization algorithm \cite{BuchbinderF18}. 
        Otherwise, $f(S_1\cap O) \le \frac{v}{2} - 2d\tau$, then
        \begin{align*}
            f(S) &\ge \frac{1}{2}f(O) -\frac{1}{2}f(S_1 \cap O) - 2d\tau \\
            &\ge \frac{1}{2}f(O) - \frac{v}{4} + d\tau - 2d\tau \\
            &\ge 2(d+1) \tau - (d+1)\tau -d\tau \\
            &= \tau.
        \end{align*}
        Where the last inequality is by $f(O)\ge v\ge 4(d+1)\tau$.
        Thus, we have $f(S) \ge \frac{\alpha}{4(d+1)} \text{OPT}$ in any case.
    \end{itemize}
\end{proof}

\subsection{Streaming Algorithm for $d$-Knapsack Knowing the Max Density}

We obtain an approximation of OPT by Algorithm \ref{alg: Streaming for d-Knapsack Knowing OPT} which require OPT as an input. That is, we have to estimate OPT first. Like in the cardinality constraint that we estimate OPT using the \emph{max value} of the element, we can estimate OPT by the \emph{max density} of the element. Let $m=\max _{1\le i \le d, 1\le j \le n} f(u_j)/c_{i,j}$, the maximum unit value.

\begin{lemma}
    Let $Q=\{(1+\epsilon)^i\mid i\in \mathbb{Z}, m/(1+\epsilon)\le (1+\epsilon)^i \le bm\}$. Then there is at least one $v\in Q$ such that $(1-\epsilon)\text{OPT}\le v \le \text{OPT}$.
    \label{le:max density}
\end{lemma}

\begin{proof}
    Without lose of generation, let $m=f(u_{j'})/c_{i',j'}$, $i' \in [1,d]$ and $j' \in [1, n]$. Since $c_{i', j'}\ge 1$, we have $\text{OPT}\ge f(u_{j'}) = mc_{i', j'}\ge m$.
    On the other side, we have $\text{OPT} = \sum_{i=1}^{|O|}( f(\{u_1, u_2, \dots, u_i\})- f(\{u_1, u_2, \dots, u_{i-1}\})) \le \sum_{i=1}^{|O|} f(u_i)\le m \sum_{i=1}^{|O|} c_{1, u_i} \le bm$.

    Let $v=(1+\epsilon)^{\lfloor \log_{1+\epsilon} \text{OPT} \rfloor}$, we obtain $\frac{m}{1+\epsilon} \le (1-\epsilon) \text{OPT} \le v\le \text{OPT} \le bm$.
\end{proof}

By Lemma \ref{le:max density}, we design the following algorithm requiring the \emph{max density} as an input.

\begin{algorithm}[ht]
	\caption{Streaming Repeat Greedy for $d$-Knapsack Knowing $m$}
	\begin{algorithmic}[1]
		\State \textbf{Input} $m$.
            \State $Q= \{(1+\epsilon)^i\mid i\in \mathbb{Z}, m/(1+\epsilon)\le (1+\epsilon)^i \le bm\}$.
            \For{$v\in Q$}
            \State $S_1^v , S_2^v , S_3^v , S^v = \emptyset$.
            \EndFor
		\For{$j=1$ to $n$}
            \Comment{$u_j$ is the $j$-th element in the dataset. } 
            \For{$v\in Q$}
            \State $\tau = \frac{v}{4(d+1)}$.
            \Comment{$\tau$ is the threshold. }  
            \If{$c_{i,j}\ge \frac{b}{2}$ and $\frac{f(u_j)}{c_{i,j}}\ge \frac{2\tau}{b}$ for some $1\le i\le d$}
            \State $S^v=\{u_j\}$.
            \ElsIf{$\frac{f(u_j\mid S_1^v)}{c_{i,j}} \ge \frac{2\tau}{b} $ and $\sum_{l\in S_1^v\cup \{u_j\}} c_{i,l} \le b $ for all $1\le i\le d$} 
            \State $S_1^v= S_1^v\cup\{u_j\}$.
            \ElsIf {$\frac{f(u_j\mid S_2^v)}{c_{i,j}} \ge \frac{2\tau}{b} $ and $\sum_{l\in S_2^v\cup \{u_j\}} c_{i,l} \le b $ for all $1\le i\le d$} 
            \State $S_2^v= S_2^v\cup\{u_j\}$.
            \EndIf
            \EndFor
            \EndFor
            \For{$v\in Q$}
            \State $S_3^v = $ \textbf{Unconstrained} ($S_1^v$).
            \EndFor
            \State\Return $S= \argmax_{v\in Q}\{f(S_1^v),f(S_2^v),f(S_3^v),f(S^v)\}$.
	\end{algorithmic}
	\label{alg: Streaming Repeat Greedy d-Knapsack Knowing m}
\end{algorithm}

\begin{theorem}
    Assuming that the input $m$ satisfies $m=\max_{1\le i\le d, 1\le j\le n} f(u_j)/c_{ij}$, Algorithm \ref{alg: Streaming Repeat Greedy d-Knapsack Knowing m} satisfies the following properties, where the unconstrained submodular maximization is solved by a deterministic $1/2$-approximation algorithm proposed by \cite{BuchbinderF18}.
    \begin{itemize}
         \item It outputs a set $S$ such that $Cx_S\le b$ and $f(S)\ge (\frac{1}{4(d+1)}-\epsilon)$ OPT.
         \item It does 1 pass over the data set, stores at most $O(\frac{b\log b}{\epsilon})$ elements and has $O(\frac{\log b}{\epsilon})$ query complexity per element.
    \end{itemize}
    \label{thm: Streaming Repeat Greedy d-Knapsack Knowing m}
\end{theorem}

\begin{proof}
    By Lemma \ref{le:max density}, we choose $v\in Q$ such that $(1-\epsilon)\text{OPT}\le v\le \text{OPT}$. Then by Theorem \ref{thm: Streaming for d-Knapsack Knowing OPT}, the output set $S$ satisfies $f(S)\ge (\frac{1}{4(d+1)}-\epsilon)\text{OPT}$.

    Notice that there are $O(\frac{\log b}{\epsilon})$ elements in $Q$, and for each $v$ there are at most $b$ elements in $S_k^v$. It is known from the background setting of the streaming algorithm that the number $n$ of ground sets $N$ is much larger than the capacity $b$. Thus, Algorithm \ref{alg: Streaming Repeat Greedy d-Knapsack Knowing m} stores at most $O(\frac{b\log b}{\epsilon})$ elements and has $O(\frac{\log b}{\epsilon})$ query complexity per element.
\end{proof}

\subsection{1-Pass Streaming Algorithm for $d$-Knapsack}
We modify the estimation candidate set $Q$ into $Q=\{(1+\epsilon)^k\mid k\in \mathbb{Z}, m/(1+\epsilon)\le (1+\epsilon)^k \le 2(d+1)bm\}$, and let $m$ be the current maximum marginal value per weight of all single element. 
The streaming algorithm will perform a parallel threshold algorithm while updating $m$ and the estimation candidate set $Q$. 
It can be seen that when a threshold $v$ is instantiated from the set $Q$, every elemnt with marginal value $\frac{v}{4(d+1)}$ to $S^v$ will appear on or after $v$ is instantiated.
Finally, we develop the final 1-pass streaming algorithm and establish the following theorem, whose proof follows the same lines as the proof of Theorem \ref{thm: Streaming Repeat Greedy d-Knapsack Knowing m}.

\begin{algorithm}[ht]
	\caption{1-Pass Streaming Repeat Greedy for $d$-Knapsack}
	\begin{algorithmic}[1]
		\State $Q=\{(1+\epsilon)^k\mid k\in \mathbb{Z}\}$.
            \For{$v\in Q$}
            \State $S_1^v , S_2^v , S_3^v , S^v = \emptyset$.
            \EndFor
            \State $m=0$.
            \For{$j=1$ to $n$}
            \Comment{$u_j$ is the $j$-th element in the dataset. } 
            \For{$i=1$ to $d$}
            \State $m=\max \{m, \frac{f(u_j)}{c_{i,j}}\}$.
            \EndFor
            \State $Q_j=\{(1+\epsilon)^k\mid k\in \mathbb{Z}, m/(1+\epsilon)\le (1+\epsilon)^k \le 2(d+1)bm\}$.
            \For{$v\in Q_j$}
            \State $\tau = \frac{v}{4(d+1)}$.
		\Comment{$\tau$ is the threshold. } 
            \If{$c_{i,j}\ge \frac{b}{2}$ and $\frac{f(u_j)}{c_{i,j}}\ge \frac{2\tau}{b}$ for some $1\le i\le d$}
            \State $S^v=\{u_j\}$.
            \ElsIf{$\frac{f(u_j\mid S_1^v)}{c_{i,j}} \ge \frac{2\tau}{b} $ and $\sum_{l\in S_1^v\cup \{u_j\}} c_{i,l} \le b $ for all $1\le i\le d$} 
            \State $S_1^v= S_1^v\cup\{u_j\}$.
            \ElsIf {$\frac{f(u_j\mid S_2^v)}{c_{i,j}} \ge \frac{2\tau}{b} $ and $\sum_{l\in S_2^v\cup \{u_j\}} c_{i,l} \le b $ for all $1\le i\le d$} 
            \State $S_2^v= S_2^v\cup\{u_j\}$.
            \EndIf
            \EndFor
            \EndFor
            \For{$v\in Q_n$}
            \State $S_3^v = $ \textbf{Unconstrained} ($S_1^v$).
            \EndFor
            \State\Return $S= \argmax_{v\in Q_n}\{f(S_1^v),f(S_2^v),f(S_3^v),f(S^v)\}$.

	\end{algorithmic}
	\label{alg: Streaming Repeat Greedy d-Knapsack}
\end{algorithm}

\begin{theorem}
    Algorithm \ref{alg: Streaming Repeat Greedy d-Knapsack} satisfies the following properties, where the unconstrained submodular maximization is solved by a deterministic $1/2$-approximation algorithm proposed by \cite{BuchbinderF18}.
    \begin{itemize}
         \item It outputs a set $S$ such that $Cx_S\le b$ and $f(S)\ge (\frac{1}{4(d+1)}-\epsilon)$ OPT.
         \item It does 1 pass over the data set, stores at most $O(\frac{b\log b}{\epsilon})$ elements and has $O(\frac{\log b}{\epsilon})$ query complexity per element.
    \end{itemize}
    \label{thm: Streaming for d-Knapsack}
\end{theorem}

\begin{proof}
    When a threshold $v$ is instantiated from the set $Q$, every element with marginal value $\frac{v}{4(d+1)}$ to $S^v$ will appear on or after $v$ is instantiated. 
    The remaining proof is similar to Theorem \ref{thm: Streaming Repeat Greedy d-Knapsack Knowing m}.
\end{proof}

One immediate corollary is that the Algorithm \ref{alg: Streaming Repeat Greedy d-Knapsack} can achieve a $\frac{1}{8}-\epsilon$ approximation when $d=1$, which means the single knapsack constraint. 

\section{Conclusion and Future Work}
\label{sec: conclusion}

In this paper, we propose deterministic streaming algorithms with improved approximation ratios for non-monotone submodular maximization under a cardinality constraint, and a $d$-knapsack constraint, respectively.

For the $d$-knapsack constraint, a question that we are more concerned about next is whether the approximation ratio of the streaming algorithm can be further improved. The current analysis on the volume of the knapsack is still a bit rough. We believe that through more detailed discussion, better results can be achieved.
For a cardinality constraint, an open question in this field is whether deterministic algorithms can achieve the same approximation ratio as randomized algorithms.
When the objective function is non-monotone, though our algorithms improve the best known deterministic algorithms, their approximation ratios are still worse than the best randomized algorithms. It is very interesting to fill these gaps.

\bibliographystyle{plain}
\bibliography{lib-reference}

\end{document}